\documentclass[conference]{IEEEtran}
\IEEEoverridecommandlockouts

\usepackage[utf8]{inputenc}
\usepackage[T1]{fontenc}
\usepackage{url}
\usepackage{ifthen}
\usepackage{cite}
\usepackage[cmex10]{amsmath}

\usepackage{xcolor}
\usepackage{mathtools}
\mathtoolsset{showonlyrefs}

\usepackage{epsfig,amsfonts,subfigure,color,amsbsy}
\usepackage[nolist]{acronym}
\usepackage{graphicx,cite,amssymb,amsthm}
\usepackage{amsmath}
\usepackage{mathrsfs}
\usepackage{tikz}
\usepackage{bm}
\usepackage{multirow}
\usepackage{bigstrut}
\usepackage{psfrag}
\usepackage{stackengine}
\usetikzlibrary{shapes,arrows}
\usepackage{xcolor}
\usepackage{mathtools}
\usepackage{tabularx,array}
\usepackage{mathbbol}
\usepackage{lipsum}

\usepackage{tikz}
\usetikzlibrary{shapes,arrows}
\usetikzlibrary{calc}
\usetikzlibrary{patterns}
\usepackage{mathrsfs}
\usetikzlibrary{fadings}
\usetikzlibrary{shadows.blur}
\usetikzlibrary{spy}
\usetikzlibrary{positioning,chains,fit}

\usepackage{xurl}
\usepackage{booktabs}
\usetikzlibrary{backgrounds,arrows.meta}
\usepackage{pgfplots}
\pgfplotsset{compat=newest}

\usepackage{orcidlink}

\newtheorem{theorem}{Theorem}
\newtheorem{lemma}[theorem]{Lemma}
\newtheorem{proposition}[theorem]{Proposition}
\newtheorem{definition}[theorem]{Definition}

\definecolor{PsZero}{RGB}{0,114,178}
\definecolor{PsFive}{RGB}{0,158,115}
\definecolor{PsTen}{RGB}{204,121,167}

\newcommand{\pr}[1]{\operatorname{Pr}\left[#1\right]}
\newcommand{\expect}[1]{\mathbb{E}\left[#1\right]}
\newcommand{\alphabetq}{\Sigma_q}
\newcommand{\pins}{p_{\mathrm{i}}}
\newcommand{\pdel}{p_{\mathrm{d}}}
\newcommand{\psub}{p_{\mathrm{s}}}
\newcommand{\ptran}{p_{\mathrm{t}}}

\usepackage{hyperref}

\begin{document}

\title{Orthogonal SC-LDPC Codes for the IDS Channel}

\author{%
  \IEEEauthorblockN{
    \textbf{Frederik~Walter}\IEEEauthorrefmark{1}\orcidlink{0009-0005-3997-3584},
    \textbf{Jessica~Bariffi}\IEEEauthorrefmark{1}\orcidlink{0000-0002-4304-2521},
    and
    \textbf{Antonia Wachter-Zeh}\IEEEauthorrefmark{1}\orcidlink{0000-0002-5174-1947}
  }

  \IEEEauthorblockA{\IEEEauthorrefmark{1}%
                    School of Computation, Information and Technology,
                    Technical University of Munich,
                    80333 Munich, Germany}
  \IEEEauthorblockA{Email:
                    frederik.walter@tum.de,
                    jessica.bariffi@tum.de,
                    antonia.wachter-zeh@tum.de}
                    \vspace{-1.5\baselineskip}
   \thanks{Funded by the European Union (DiDAX, 101115134). Views and opinions expressed are however those of the author(s) only and do not necessarily reflect those of the European Union or the European Research Council Executive Agency. Neither the European Union nor the granting authority can be held responsible for them.}%
}

\maketitle

\begin{abstract}
    We propose a coding scheme for the insertion, deletion and substitution (IDS) channel tailored to the short strands of DNA-based data storage. Instead of an inner synchronisation code, a spatially coupled LDPC (SC-LDPC) code is placed across the strands, with its coupling axis aligned to the position along the strands. Every decoded position thus becomes an anchor to synchronise the neighbouring strands. The decoding waves of the SC-LDPC code carry these anchors from both ends of the strands towards the middle.
    We bound the achievable rate from above by the information rate of a single strand and from below by that of a separated receiver. For strands of 100 nucleotides (the relevant regime for DNA-based data storage), our scheme outperforms a rate-matched concatenated construction.
\end{abstract}


\section{Introduction}
\label{sec:introduction}

DNA is a candidate medium for archival data storage because of its density and durability~\cite{church_2012_NextGenerationDigital,goldman_2013_PracticalHighcapacity,erlich_2017_DNAFountain}. Data is written by synthesising a pool of short strands, today at most a few hundred nucleotides long~\cite{lin_2025_BlockLength}, and read back by sequencing them. Both steps introduce insertions, deletions and substitutions independently in every strand~\cite{heckel_2019_CharacterizationDNA}, so the storage channel is a collection of independent insertion, deletion and substitution (IDS) channels, one per strand. Most coding schemes combine an inner code that restores synchronisation along each strand with an outer code, typically across strands, that corrects the residual errors~\cite{press_2020_HEDGESErrorcorrecting,welter_2026_EndtoEndCoding,zhang_2025_DesignAnalysis}.

The inner layer costs rate and blocklength. Markers or watermarks~\cite{davey_2001_ReliableCommunication,wang_2011_SymbolLevelSynchronization,haghighat_2025_HalfMarkerCodes}, guard bands~\cite{pfister_2021_PolarCodes,tal_2022_PolarCodes} and inner synchronisation codes~\cite{maarouf_2023_ConcatenatedCodes} all reduce the rate left for the outer code. Moreover, a per-strand codeword can be no longer than its strand, so it pays a finite-length penalty at the short strand lengths that synthesis delivers~\cite{maarouf_2023_FiniteBlocklength}.

We use no inner code and synchronise through the outer code instead. Every column of a frame of $n$ strands is protected by parity checks shared by $n$ strands that drift independently, so a column the code certifies becomes a synchronisation anchor for all of them at once. This orthogonal layout appears in Geno-Weaving~\cite{wang_2025_GenoWeavingFramework,lin_2025_BlockLength} for deletion channels and in the salami-slicing trellis~\cite{wu_2026_SalamiSlicing}, which advances decision-feedback trellises one position at a time, polar-decodes each cross-strand slice and is evaluated at low error rates.
Delayed coding~\cite{shibata_2022_DelayedCoding,shibata_2024_HierarchicalInterleaving} also uses decoded codewords as pilots, but for a single long binary sequence and with an asymptotic analysis. Our scheme couples the columns by a spatially coupled low-density parity-check (SC-LDPC) code, feeds soft information back to the trellises, can revoke a decision, and operates on short quaternary strands at substantially higher insertion and deletion rates.

\emph{Contributions:} We propose an SC-LDPC code across strands whose spatial position is the position along the strand, so that the coupling axis and the drift axis coincide (Section~\ref{sec:structure:construction}). We decode it by alternating between a beam-pruned BCJR on the drift trellis of each strand and windowed belief propagation on the two bit-planes. The wave-like decoding behaviour of SC-LDPC codes is utilised to move the anchored ends towards each other (Sections~\ref{sec:trellis} and~\ref{sec:schedule}).
At the actual strand length, we bound the rate from above by the information rate of a single strand with both ends known, and from below by the rate of a separated receiver that synchronises each strand once and then decodes across strands (Section~\ref{sec:rates}). In simulations (Section~\ref{sec:results}), the scheme decodes at higher insertion and deletion rates than a rate-matched concatenated construction~\cite{maarouf_2023_ConcatenatedCodes} at every substitution rate, beyond the finite-blocklength benchmark of a single-strand code of comparable length~\cite{maarouf_2023_FiniteBlocklength}, and well above the rate of the separated receiver.

\section{Channel Model and Achievable Rates}
\label{sec:channel}

In the following we denote by $\alphabetq := \{0, \ldots, q-1\}$ the $q$-ary alphabet for a positive integer $q$. Vectors are denoted by bold lowercase letters, and random variables of
vectors and matrices by bold uppercase letters.
For random variables $X,Y$, we denote the entropy by $H(X)$ and the conditional entropy of $X$ given $Y$ by $H(X\mid Y)$.
\subsection{The Insertion, Deletion and Substitution Channel}

We recap the state-based IDS channel of
Davey and MacKay~\cite{davey_2001_ReliableCommunication} over an alphabet of
size $q$.
The IDS channel maps an input sequence $\mathbf{x} \in \alphabetq^L$ to an output $\mathbf{y} \in \alphabetq^M$ of random length $M$, reading $\mathbf{x}$ symbol by symbol. At each time step $k$, the channel applies one of three actions, independently of $\mathbf{x}$ and of all previous actions: with probability $\pins$ a uniformly random symbol is appended to the output and the input symbol $x_k$ is left unread; with probability $\pdel$ the symbol $x_k$ is deleted; and with probability $\ptran = 1 - \pins - \pdel$ it is transmitted, in which case it is replaced with probability $\psub$ by one of the other $q-1$ symbols, chosen uniformly. The channel stops once all symbols of $\mathbf{x}$ have been read. Each input symbol is thus preceded by a geometrically distributed number of insertions and then either deleted or transmitted. The receiver only observes the output $\mathbf{y}$ and knows neither the actions applied nor the alignment of the output and input symbols.
An illustration of the IDS channel can be found in \cite{davey_2001_ReliableCommunication,maarouf_2023_ConcatenatedCodes}.
The natural synchronisation variable is the \emph{drift}.
\begin{definition}\label{def:drift}
    For $k \in \{0, \ldots , L \}$, let $j_k$ denote the number of output symbols emitted up to and including the action that reads $x_k$, with $j_0 = 0$. We define the \emph{drift} after $k$ input symbols by $\delta_k := j_k - k$.
\end{definition}
Note that deletions decrement the drift, insertions increment it, and transmissions leave it unchanged.
Then, $(\delta_k)_{k \in \{0, \ldots, L\}}$ is a Markov chain with $\delta_0=0$ and $\delta_L=M-L$, independent of $\mathbf{x}$.
If $x_k$ is transmitted, it is received as $y_{k+\delta_k}$.

A frame of $n$ strands uses this channel $n$ times independently, and each strand is read once. We assume that the receiver knows the ordering of the strands, as provided in practice by an index, and that no strand is lost or duplicated.
Following~\cite{maarouf_2023_ConcatenatedCodes}, we sweep $\pins = \pdel$ and treat $\psub$ as a separate parameter; the decoder makes no use of this split.

\subsection{Achievable Rates}
\label{sec:rates}

The capacity of the IDS channel is unknown. However, lower and upper bounds are given in~\cite{gallager_1961_SequentialDecoding,fertonani_2011_BoundsCapacity}; for the pure deletion channel, see~\cite{mitzenmacher_2009_SurveyResults,rubinstein_2023_ImprovedUpper,morozov_2025_SimpleConverse}.
Every scheme in this paper uses uniform i.i.d.\ inputs, so the relevant benchmarks are information rates for uniform inputs. Since our strands are short, we define them at the level of a strand. In the following, if nothing else is stated, $\mathbf{X}$ is uniformly distributed over $\alphabetq^L$. We denote by $\mathbf{Y}$ the output of the IDS channel given $\mathbf{X}$.

\begin{definition}
    For an integer $L\geq 1$,
    define the \emph{finite-$L$ rate} as
    \vspace{-6pt}
    \begin{equation}
      I_L = \tfrac{1}{L}\, I(\mathbf{X};\mathbf{Y}) = \tfrac{1}{L}\, \bigl[H(\mathbf{Y}) - H(\mathbf{Y}\mid\mathbf{X})\bigr].
      \label{eq:il}
    \end{equation}
    We further define the symmetric information rate by
        $I_{\mathrm{s}} = \lim_{L\to\infty} I_L$.
    The existence of the limit is proven in \cite{tal_2022_PolarCodes}.
\end{definition}

$I_L$ is thus the relevant ceiling for codes across strands, with both anchors $\delta_0=0$ and $\delta_L=M-L$ part of the observation.

\begin{lemma}\label{lemma:H(Y)}
    The
    length $M$ of the IDS channel's output $\mathbf{Y}$ is independent of the input $\mathbf{X}$. Furthermore, given the length $M=m$, the output $\mathbf{Y}$ is uniformly distributed over $\alphabetq^m$ and it holds that
    \vspace{-6pt}
    \begin{align}
        H(\mathbf{Y}) &= H(M) + \expect{M}\log_2 q,
        \\
        I_L &= \tfrac{1}{L}\Bigl(\expect{M}\log_2 q
        - H(\mathbf{Y}\mid\mathbf{X},M)\Bigr).
        \vspace{-6pt}
    \end{align}
\end{lemma}
\begin{proof}
    Note that $M$ depends solely on the action sequence performed on $\mathbf{X}$ and not on $\mathbf{X}$ itself. Thus, every output symbol derives from either its own insertion, which is uniformly distributed, or its own input symbol, which, again, is uniformly distributed. Hence, every output symbol follows a uniform distribution. For a given length $M=m$ and a sequence of actions, the output $\mathbf{Y}$ is uniform over $\alphabetq^m$.
    Since $M$ is a function of $\mathbf{Y}$, we obtain
    $$
        H(\mathbf{Y})
        = H(M)+H(\mathbf{Y}\mid M)
        = H(M)+\expect{M}\log_2 q.
    $$
    Since $M$ is independent of $\mathbf{X}$, we additionally observe that $H(\mathbf{Y}\mid\mathbf{X})=H(M)+H(\mathbf{Y}\mid\mathbf{X},M)$.
\end{proof}

We estimate $H(\mathbf{Y}\mid\mathbf{X}) = -\expect{\log_2 \pr{\mathbf{Y}\mid\mathbf{X}}}$ by the
Arnold--Loeliger method~\cite{arnold_2001_InformationRate,arnold_2006_SimulationBasedComputation},
as applied to insertion and deletion channels
in~\cite{hu_2010_AchievableInformation,maarouf_2023_ConcatenatedCodes}:
$\pr{\mathbf{Y}\mid\mathbf{X}}$ is read off the terminal state $\delta_L = M - L$ of the
forward recursion of the drift trellis with the transmitted strand known.
$H(\mathbf{Y})$ follows exactly from Lemma~\ref{lemma:H(Y)}; its term $H(M)/L$, of order $(\log L)/L$, also appears in $H(\mathbf{Y}\mid\mathbf{X})$ and must be kept on both sides, since it is not negligible at $L = 100$. The two known anchors are side information, and their amount grows more slowly than $L$. Per symbol it therefore shrinks as $L$ grows,
so $I_L$ decreases towards the symmetric information rate $I_{\mathrm{s}}$ from above.
The floor of the rate ladder is the \emph{separated} receiver, which synchronises every strand once and then decodes across strands without feedback.
\begin{definition}\label{def:GMI}
    Consider $q = 4$.
    With the binary labelling $X_k = 2B_{k,0}+B_{k,1}$, the bit $B_{k,i}$, $i\in\{0,1\}$, is the bit of position $k\in\{1,\dots,L\}$ in bit-plane $i$. We denote by
    $\lambda_{k,i} = \ln\bigl(\pr{B_{k,i} = 0 \mid \mathbf{Y}} / \pr{B_{k,i} = 1 \mid \mathbf{Y}}\bigr)$
    the posterior log-likelihood ratio (LLR) of bit $B_{k,i}$, computed by a single full-strand BCJR pass with uniform priors and both anchors. The \emph{BCJR-once generalised mutual information (GMI)}, in bits per nucleotide, is defined by
    \vspace{-6pt}
    \begin{equation}
        \mathrm{GMI}_L \!
        = 2 - \frac{1}{L}\sum_{k=1}^{L}\sum_{i\in\{0,1\}}
        \!\!\!
        \expect{\log_2\bigl(1 + \mathrm{e}^{-(1-2B_{k,i})\lambda_{k,i}}\bigr)} \! ,
        \label{eq:gmi}
        \vspace{-6pt}
    \end{equation}
    where the expectation is over $\mathbf{X}$ and the channel realisation.
\end{definition}

By~\cite{martinez_2009_BitInterleavedCoded}, $L \cdot \mathrm{GMI}_L$ is the generalised mutual information of the bit metric $\prod_{k,i}\pr{B_{k,i}\mid\mathbf{Y}}$, which treats all bits as independent. Hence $\mathrm{GMI}_L$ is achievable by an outer code across strands whose decoder uses the posteriors of a single BCJR pass per strand, i.e.\ by the separated receiver. It is therefore an achievable rate for that receiver, not an upper bound: a receiver that feeds the beliefs of the outer decoder back to the BCJR as priors can exceed it, and Section~\ref{sec:localrates} shows by how much.

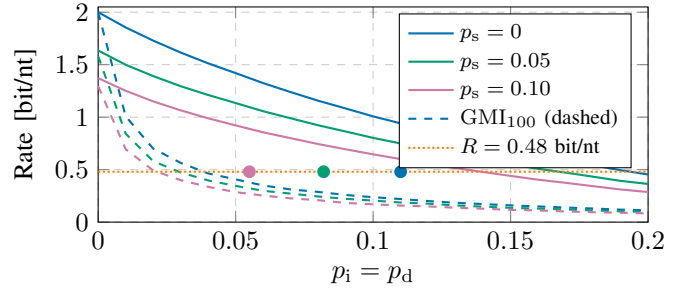
\begin{figure}[t]
    \centering
    \begin{tikzpicture}
\begin{axis}[
    width=\columnwidth, height=0.5\columnwidth,
    xlabel={$p_{\mathrm{i}} = p_{\mathrm{d}}$},
    ylabel={Rate [bit/nt]},
    xmin=0, xmax=0.20,
    ymin=0, ymax=2.05,
    grid=major,
    grid style={dashed, gray!40},
    xtick={0, 0.05, 0.10, 0.15, 0.20},
    xticklabel style={/pgf/number format/fixed},
    legend pos=north east,
    legend cell align={left},
    legend style={font=\footnotesize},
    unbounded coords=discard,
    every axis plot/.append style={thick, no markers},
]
    \addplot[color=PsZero]
        table[x=p, y=air_ps000, col sep=comma]{fig/rate_ladder_4ary_L100.csv};
    \addlegendentry{$p_{\mathrm{s}} = 0$}
    \addplot[color=PsFive]
        table[x=p, y=air_ps050, col sep=comma]{fig/rate_ladder_4ary_L100.csv};
    \addlegendentry{$p_{\mathrm{s}} = 0.05$}
    \addplot[color=PsTen]
        table[x=p, y=air_ps100, col sep=comma]{fig/rate_ladder_4ary_L100.csv};
    \addlegendentry{$p_{\mathrm{s}} = 0.10$}
    \addplot[color=PsZero, dashed]
        table[x=p, y=gmi_ps000, col sep=comma]{fig/rate_ladder_4ary_L100.csv};
    \addlegendentry{$\mathrm{GMI}_{100}$ (dashed)}
    \addplot[color=PsFive, dashed, forget plot]
        table[x=p, y=gmi_ps050, col sep=comma]{fig/rate_ladder_4ary_L100.csv};
    \addplot[color=PsTen, dashed, forget plot]
        table[x=p, y=gmi_ps100, col sep=comma]{fig/rate_ladder_4ary_L100.csv};
    \addplot[color=orange, densely dotted] coordinates {(0, 0.48) (0.20, 0.48)};
    \addlegendentry{$R = 0.48$ bit/nt}
    \addplot[only marks, mark=*, mark size=2pt, color=PsZero, forget plot] coordinates {(0.110, 0.48)};
    \addplot[only marks, mark=*, mark size=2pt, color=PsFive, forget plot] coordinates {(0.082, 0.48)};
    \addplot[only marks, mark=*, mark size=2pt, color=PsTen, forget plot] coordinates {(0.055, 0.48)};
\end{axis}
\end{tikzpicture}
    \vspace{-5mm}
    \caption{Rate ladder for $L = 100$ versus $\pins = \pdel$:
    $I_{100}$~\eqref{eq:il} (solid, $10^4$ strands per point),
    $\mathrm{GMI}_{100}$~\eqref{eq:gmi} (dashed, $2000$ strands), the code
    rate $R = 0.48$ bit/nt (dotted), and the waterfalls of the inter-strand
    scheme at a frame error rate (FER) of $10^{-3}$ (dots).}
    \vspace{-5mm}
    \label{fig:rates}
\end{figure}

Fig.~\ref{fig:rates} shows both rungs for $q = 4$ and $L = 100$. The information rate $I_{100}$ stays above our code rate $R=0.48$~bit/nt up to $\pins=\pdel\approx 0.194$, $0.168$ and $0.139$ for $p_\mathrm{s}=0$, $0.05$ and $0.10$, whereas $\mathrm{GMI}_{100}$ falls below it already at $0.038$, $0.029$ and $0.021$: with uniform priors, the position ambiguity of a 100-nucleotide strand costs most of the rate.
A separated receiver could therefore code across strands at our rate only at insertion and deletion rates almost three times smaller than those at which the inter-strand scheme operates (Section~\ref{sec:results}).

\subsection{Local Rates and Achievability}\label{sec:localrates}
The two rungs in Section~\ref{sec:rates} differ in what the receiver knows once it has decided on a position. Assume that positions $1, \ldots, k-1$, for some $k$, are known. Hence, the anchor has moved from position $0$ to position $k-1$. We define a local rate for position $k$ as follows.
\begin{definition}\label{def:local-rate}
    Let $\mathbf{X} = (X_1, \ldots, X_L)$ be uniformly distributed over $\Sigma_4^L$, and let $\mathbf{Y}$ denote the output of the IDS channel given $\mathbf{X}$.
    For any $k \in \{1, \ldots, L\}$ the \emph{local rate} of position $k$ is defined by
    $
        r_k := I(X_k ; \mathbf{Y}  \mid  X_1, \ldots, X_{k-1}).
    $
\end{definition}
\noindent By the chain rule for mutual information, we have that
\vspace{-7pt}
\begin{align}
    \sum_{k = 1}^{L} r_k = I(\mathbf{X} ; \mathbf{Y}) = L \cdot I_L,
    \vspace{-10pt}
\end{align}
which Proposition~\ref{prop:codes-closing-gap} below turns into an achievable rate for codes across strands. Lemma~\ref{lemma:gap_GMI-I} shows that $\mathrm{GMI}_L$ never exceeds $I_L$. The difference is what anchors can recover.
\begin{lemma}\label{lemma:gap_GMI-I}
    Let $\mathbf{X} = (X_1, \ldots, X_L)$ be uniformly distributed over $\Sigma_4^L$, and let $\mathbf{Y}$ denote the output of the IDS channel given $\mathbf{X}$. Using the binary labelling $X_k = 2B_{k, 0} + B_{k, 1}$, we have
    $
    \mathrm{GMI}_L \leq I_L .
    $
\end{lemma}
\begin{proof}
    Since $\lambda_{k,i}$ is the posterior LLR, each expectation in~\eqref{eq:gmi} equals $H(B_{k,i}\mid\mathbf{Y})$, so, as $H(B_{k,i})=1$, $\mathrm{GMI}_L=\tfrac{1}{L}\sum_k\sum_i I(B_{k,i};\mathbf{Y})$. As the bits of $X_k$ are independent, $\sum_i I(B_{k,i};\mathbf{Y})\le I(X_k;\mathbf{Y})$, and as the positions are independent, $I(X_k;\mathbf{Y})\le r_k$. Summing over $k$ and using the chain rule completes the proof.
\end{proof}
The gap between the two rungs is thus what knowledge of earlier positions, i.e., anchors, contributes. Codes across strands can recover this gap, as the following result shows.
\begin{proposition}\label{prop:codes-closing-gap}
    Consider a frame of $n$ strands. Let each column $\mathbf{x}_k$ be a codeword of a code $\mathcal{C}_k$ of rate $R_k < r_k$. Assume the columns are decoded in increasing order of their indices, each given the knowledge of the previously decoded columns. As $n$ grows large, there exist codes achieving vanishing error probability and $\tfrac{1}{L} \sum_{k}R_k \longrightarrow I_L$.
\end{proposition}
\begin{proof}
    With random codes of uniform symbols, and given correctly decoded columns $1,\dots,k-1$, each strand is an independent use of a channel with input $X_k$ and output $(\mathbf{Y},X_1,\dots,X_{k-1})$, whose mutual information is $r_k$. Hence every rate $R_k<r_k$ is achievable at stage $k$, and a union bound over the $L$ stages completes the proof.
\end{proof}

\section{Inter-Strand LDPC Code}
\label{sec:structure}

\subsection{Construction}\label{sec:structure:construction}
A frame consists of $n$ strands of $L$ nucleotides each, arranged as a matrix $\mathbf{C} \in \Sigma_4^{n \times L}$ whose rows are the strands and whose $t$-th column $\mathbf{x}_t$ collects position $t$ of all strands.
With the binary labelling $x=2b_0+b_1$, the frame splits into two binary planes $\mathbf{C}^{(0)}, \mathbf{C}^{(1)}\in\mathbb{F}_2^{n\times L}$, and we denote the $t$-th column of $\mathbf{C}^{(i)}$ by $\mathbf{x}_{t}^{(i)}$. Each plane is a codeword of the same binary SC-LDPC code~\cite{kudekar_2011_ThresholdSaturationa}. No inner code and no markers are used. All redundancy is in this single outer code.

We use a memory-one quasi-cyclic SC-LDPC code built from a protograph. The code is defined by two base matrices $\mathbf{P}_0,\mathbf{P}_1$ of size $m_\mathrm{b}\times n_\mathrm{b}$, lifted with $Z \times Z$ circulant permutation matrices to $\mathbf{H}_0,\mathbf{H}_1\in\mathbb{F}_2^{m_\mathrm{b}Z\times n_\mathrm{b}Z}$. We set $n=n_\mathrm{b}Z$, so that the
$n$ variable nodes of spatial position $t$ are the bits of column $t$ (Fig.~\ref{fig:sc-structure}).
The code is terminated: with $\mathbf{x}_0 = \mathbf{x}_{L+1} = \mathbf{0}$, a plane is a codeword if and only if
\begin{equation}
    \mathbf{H}_0\,\mathbf{x}_t \oplus \mathbf{H}_1\,\mathbf{x}_{t-1} = \mathbf{0},
    \qquad t = 1, \dots, L+1 .
    \label{eq:sc-syndrome}
\end{equation}
We refer to the checks of index $t$ in~\eqref{eq:sc-syndrome} as \emph{block row $t$}.
The terminating block row $L+1$ is satisfied by solving for the information bits of the last $T$ positions, which gives the rate \(R_{\mathrm{b}} = (1 - T/L)(1 - m_{\mathrm{b}}/n_{\mathrm{b}})\) per plane, i.e., $R = 2R_\mathrm{b}$~bit/nt for the frame.

\begin{figure}[t]
    \centering
    \resizebox{\columnwidth}{!}{\begin{tikzpicture}[
    vn/.style={draw=black, circle, minimum size=0.3cm, fill=white, inner sep=0pt},
    cn/.style={draw=black, rectangle, minimum size=0.35cm, fill=white, inner sep=0pt},
    cn boundary/.style={draw=black, rectangle, minimum size=0.35cm, fill=red, inner sep=0pt},
    pics/perm3d/.style={code={
        \fill[gray!10] (0.25,-0.25) -- (1.25,1) -- (1.25,1.5) -- (0.25,0.25) -- cycle;
        \draw[draw=black] (0.25,-0.25) -- (1.25,1);
        \draw[draw=black] (0.25,0.25) -- (1.25,1.5);
        \draw[draw=black, dotted] (1.25,1) -- (1.5,1.3125);
        \draw[draw=black, dotted] (1.25,1.5) -- (1.5,1.8125);
        \draw[draw=black, fill=gray!10] (-0.25,-0.25) rectangle (0.25,0.25);
        \node at (0,0) {$\pi$};
        \fill[gray!10] (0.75,1.5) -- (-0.25,0.25) -- (0.25,0.25) -- (1.25,1.5) -- cycle;
        \draw[draw=black] (-0.25,0.25) -- (0.75,1.5);
        \draw[draw=black, dotted] (0.75,1.5) -- (1,1.8125);
    }},
    ghost/.style={opacity=0.3},
    conn_line/.style={draw=black, thick},
]

\def\xspacing{1.8}      
\def\xstrandshift{0.5}  
\def\ystrandshift{0.625}  
\def\nstrands{3}        

    
    

\foreach \pos in {1,...,6} {
    \begin{scope}[xshift={\pos*\xspacing cm}]
        \pic (perm-\pos-upper) at (0.5, -0.4) {perm3d};
        
        \foreach \s in {1,...,\nstrands} {
            \node[vn, xshift=0.1cm, yshift=0.1cm] (vn-\pos-\s) at ({\s*\xstrandshift}, {\s*\ystrandshift}) {};
            \draw (vn-\pos-\s) -- ++(-0.1,-0.7cm);
            \draw (vn-\pos-\s) -- ++(0.1,-0.7cm);
        }
        \node[rotate=51.3] at ({(\nstrands+1)*\xstrandshift}, {(\nstrands+1)*\ystrandshift}) {$\cdots$};
        
        \pic (perm-\pos-lower) at (0.5, -2) {perm3d};

        \foreach \s in {1,2} {
            \node[cn, xshift=0.05cm, yshift=-3.6cm] (cn-\pos-\s) at ({\s*\xstrandshift*1.2}, {\s*\ystrandshift*1.2}) {};
        }

        \draw (cn-\pos-1) -- ++(-0.1,0.6cm);
        \draw (cn-\pos-1) -- ++(0.1,0.6cm);
        \draw (cn-\pos-2) -- ++(-0.06,0.4cm);
        \draw (cn-\pos-2) -- ++(0.1,0.6cm);

        \node[anchor=north] (sp-\pos) at (0.5cm, -3.3cm) {\pos};
    \end{scope}
}

\begin{scope}[xshift={0 cm}, ghost]
    \foreach \s in {1,...,\nstrands} {
        \node[vn, xshift=0.1cm, yshift=0.1cm] (vn-0-\s) at ({\s*\xstrandshift}, {\s*\ystrandshift}) {};
        \draw (vn-0-\s) -- ++(-0.1,-0.7cm);
        \draw (vn-0-\s) -- ++(0.1,-0.7cm);
    }
    \pic (perm-0-upper) at (0.5, -0.4) {perm3d};
\end{scope}

\begin{scope}[xshift={7*\xspacing cm}]   
    \pic (perm-7-lower) at (0.5, -2) {perm3d};

    \foreach \s in {1,2} {
        \node[cn, xshift=0.05cm, yshift=-3.6cm] (cn-7-\s) at ({\s*\xstrandshift*1.2}, {\s*\ystrandshift*1.2}) {};
    }
    \draw (cn-7-1) -- ++(-0.1,0.6cm);
    \draw (cn-7-1) -- ++(0.1,0.6cm);
    \draw (cn-7-2) -- ++(-0.06,0.4cm);
    \draw (cn-7-2) -- ++(0.1,0.6cm);
    \node[anchor=north] at (0.5cm, -3.3cm) {7};
\end{scope}

\node[left=of sp-1] {Spatial Position:};

\foreach \pos in {1,...,6} {
    \draw[conn_line] ({\pos*\xspacing + 0.7}, -0.65cm) -- ({\pos*\xspacing + 0.6}, -1.63);
    \draw[conn_line] ({\pos*\xspacing + 0.9}, -0.45cm) -- ({\pos*\xspacing + 1.1}, -1);
    \draw[conn_line] ({\pos*\xspacing + 1.6}, 0.4cm) -- ({\pos*\xspacing + 1.45}, -0.5);
}

\foreach \pos in {0,...,6} {
    \draw[conn_line] ({\pos*\xspacing + 0.8}, -0.6) -- ({(\pos+1)*\xspacing + 0.75}, -1.4);
    \begin{scope}[on background layer]
        \draw[conn_line] ({\pos*\xspacing + 1.28}, 0) -- ({(\pos+1)*\xspacing + 1.4}, -0.6);
    \end{scope}
}

\begin{scope}
    \node[fit=(vn-1-1)(vn-6-1), fill=violet, opacity=0.2, rounded corners, inner sep=3pt] (strand1) {};
    \node[fit=(vn-1-2)(vn-6-2), fill=green, opacity=0.2, rounded corners, inner sep=3pt] (strand2) {};
    \node[fit=(vn-1-3)(vn-6-3), fill=orange, opacity=0.2, rounded corners, inner sep=3pt] (strand3) {};
\end{scope}

\end{tikzpicture}}
    \caption{Inter-strand SC-LDPC code for $L = 6$. Each coloured band is a
    strand, and the circles above position $t$ are the bits of column $t$.
    The layers $\pi$ connect them via $\mathbf{H}_0$ to block row $t$ and via
    $\mathbf{H}_1$ to block row $t+1$ (squares). Faded: the known
    $\mathbf{x}_0 = \mathbf{0}$.}
    \label{fig:sc-structure}
\end{figure}
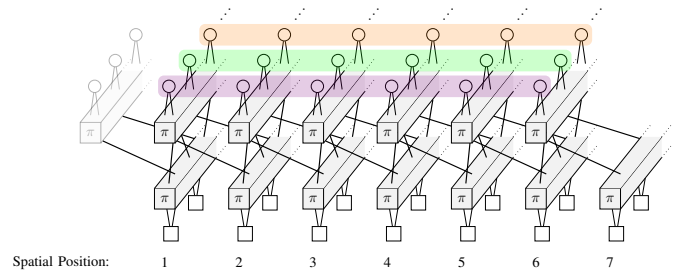

The channel and the code act along orthogonal axes.
The strands drift independently, so the timing uncertainty of one strand is unrelated to that of any other, whereas the checks of block row $t$ involve positions $t-1$ and $t$ of \emph{all} $n$ strands. Once the checks around a position are satisfied, that position is known in every strand, and it becomes a synchronisation anchor for all $n$ inner decoders at once. Since spatial position and strand position coincide, the decoding wave of the coupled code and the propagation of these anchors travel along the same axis. The decoder in Section~\ref{sec:schedule} is built on exactly this. It alternates in turbo fashion between an inner decoder per strand (Section~\ref{sec:trellis}) and belief propagation (BP) on the two planes.

\subsection{Drift Trellis and Branch Metrics}
\label{sec:trellis}

Each strand is decoded by a BCJR algorithm on the drift
trellis~\cite{davey_2001_ReliableCommunication}, whose state after step $k$ is
the drift $\delta_k$ of Section~\ref{sec:channel}. Since the received length
$M_s$ of strand $s$ is known, the trellis is anchored at both ends,
$\delta_0 = 0$ and $\delta_L = M_s - L$. The received symbols enter through the
likelihoods $\mathcal{L}_j(v) = \pr{y_j \mid v}$, which for the channel of
Section~\ref{sec:channel} are $1 - \psub$ for $v = y_j$ and
$\psub/3$ otherwise. Substitutions are thus absorbed into the
observation, and the trellis itself depends only on $\pins$ and
$\pdel$; soft basecaller output fits the same interface. The outer
decoder supplies bit LLRs $\lambda^{\mathrm{a}}_{k,i} = \ln\frac{\pr{b_{k,i}=0}}{\pr{b_{k,i}=1}}$ for
the two planes $i \in \{0,1\}$, which combine into the symbol prior
$\rho_k(v) \propto \exp\bigl(\sum_i (1 - b_i(v))\,\lambda^{\mathrm{a}}_{k,i}\bigr)$, where $b_i(v)$ is bit $i$ of symbol $v$. From
state $\delta$ at step $k$, the three branch metrics are
\begin{equation}
  \begin{aligned}
    \gamma^{\mathrm{T}}_k(\delta, v) &= \ptran\,\mathcal{L}_{k+\delta}(v)\,\rho_k(v),
      && \delta \to \delta, \\
    \gamma^{\mathrm{D}}_k &= \pdel \textstyle\sum_v \rho_k(v),
      && \delta \to \delta - 1, \\
    \gamma^{\mathrm{I}}_k(\delta) &= \tfrac{\pins}{4} \textstyle\sum_v \mathcal{L}_{k+\delta}(v),
      && \delta \to \delta + 1,
  \end{aligned}
  \label{eq:branch}
\end{equation}
for transmission, deletion and insertion. An insertion consumes a received
symbol without advancing $k$ and, as in the channel model, precedes the
transmission or deletion of $x_k$. Let $\alpha_k$ and $\beta_k$ denote the forward
and backward state metrics, and $\tilde\alpha_{k-1}$ the forward metric after
the insertions that precede step $k$. Only transmission branches distinguish the
values of $x_k$, so the extrinsic LLR of plane $i$ is
\begin{equation}
  \lambda^{\mathrm{e}}_{k,i} =
  \ln \frac{\sum_{\delta}\sum_{v:\,b_i(v)=0}
      \tilde\alpha_{k-1}(\delta)\,\gamma^{\mathrm{T}}_k(\delta,v)\,\beta_k(\delta)}
    {\sum_{\delta}\sum_{v:\,b_i(v)=1}
      \tilde\alpha_{k-1}(\delta)\,\gamma^{\mathrm{T}}_k(\delta,v)\,\beta_k(\delta)}
  - \lambda^{\mathrm{a}}_{k,i}.
  \label{eq:extrinsic}
\end{equation}
The drift is unbounded, so at each step we keep only the $\kappa$ drift values with the largest metric and discard the rest. We call this set the \emph{beam} and $\kappa$ its \emph{beam width}. A beam grown in one direction is accurate near its own anchor and diffuse near the other end, and for a moderate beam width $\kappa$ it loses the true drift there. We therefore let the recursion that starts from the anchor nearer to the positions being decoded select the states, and restrict the opposite recursion to them. When a whole strand is processed, both directions select states and the recursions are recomputed on the union. Since at most $\kappa$ drift values survive at each of the $L$ positions, one pass costs $O(L\kappa)$ per strand.

\subsection{Windowed Schedule and Anchor Propagation}
\label{sec:schedule}
The drift uncertainty is smallest at the two anchored ends of every strand and grows towards the middle. We therefore decode with two windows of $W$ positions, the left and the right \emph{wave}, that start at the two ends of the chain and move towards each other. A single full-strand BCJR pass with uniform priors initialises the state metrics, which serve as boundary conditions on the side of each window facing away from its anchor.

One decoding step is one turbo iteration restricted to the windows: the BCJR of every strand computes~\eqref{eq:extrinsic} on the window, windowed BP runs a bounded number of iterations on the block rows touching it, and the extrinsic BP output becomes the prior of the next BCJR pass. Let $\hat{\mathbf{x}}^{(i)}_t$ denote the hard decision on column $t$ of plane $i$.
\begin{definition}\label{def:fix}
    For a lookahead $\Delta\geq 1$, the left wave fixes position $t$ if positions $1,\dots,t-1$ are fixed and block rows $t,\dots,t+\Delta$ of~\eqref{eq:sc-syndrome} are satisfied by $\hat{\mathbf{x}}^{(i)}$ for both $i\in\{0,1\}$. The right wave is defined symmetrically.
\end{definition}
A fixed position is clamped in BP and enters the BCJR of every strand with priors of magnitude $\lambda_{\max}$. A run of known symbols concentrates the drift posterior, so the fixed positions act as a moving anchor for all $n$ strands, and each wave faces the drift uncertainty of $W$ positions rather than of the remaining strand. Satisfied checks do not guarantee correct decisions, however.
For a fixed plane $i$, let $\mathbf{e}_t=\hat{\mathbf{x}}^{(i)}_t\oplus\mathbf{x}^{(i)}_t$ denote the error columns.
Since the transmitted plane satisfies~\eqref{eq:sc-syndrome}, block row $t$ is satisfied by the hard decisions if and only if the error columns satisfy $\mathbf{H}_0\mathbf{e}_t\oplus\mathbf{H}_1\mathbf{e}_{t-1}=\mathbf{0}$.
Therefore, a nontrivial error $\mathbf{e}_t$ in the kernel of $\mathbf{H}_0$ that continues over the lookahead passes the rule. Hence, a larger $\Delta$ makes a false fix less likely but cannot exclude it, so fixing must be revocable. A wave that fixes no new position in $S$ consecutive steps retires. If both waves have retired and a gap remains, the decoder releases the last $\Delta+1$ fixed positions on each side of the gap, reruns a full-strand BCJR for every strand with the current priors, and restarts both waves at the edges of the gap; the release depth grows with every attempt, and the number of attempts is capped. Remaining unfixed positions are decoded by one full-strand pass followed by BP over the whole chain.

A step costs $O(nWB)$ trellis operations and $O(nW)$ BP messages per iteration, so decoding is linear in $L$ while the waves advance and quadratic in the worst case.
Unlike windowed decoding along a single sequence~\cite{shibata_2022_WindowedbasedSynchronization}, the window here moves along all $n$ strands at once, and each anchor it produces is shared by all of them.

\section{Simulation Results}
\label{sec:results}

\subsection{Simulation Parameters and Setup}
We split a regular base matrix with variable node degree $3$, check node degree $4$, $n_\mathrm{b}=16$ and $m_\mathrm{b}=12$ into the two components $\mathbf{P}_0$ and $\mathbf{P}_1$ and lift it with $Z=24$. A frame thus consists of $n=384$ strands of $L=100$ nucleotides, and with $T=4$ the rate is $R=0.48$~bit/nt, slightly below the $1/2$~bit/nt of the baselines. 
The decoder uses $W=9$, $B=70$, ten BP iterations per step, $\lambda_{\max}=30$, $S=30$ and $\Delta=2$. A frame error is declared if any of the 38\,400 nucleotides of the frame is wrong after decoding. Every point in Fig.~\ref{fig:fer4} is based on at least 200 frames, and every point at or below $\mathrm{FER}=10^{-3}$ on $10^4$ frames, i.e., $3.8\cdot10^8$ nucleotides. We define the waterfall of a curve as the value of $p_\mathrm{i}=p_\mathrm{d}$ at which it crosses $\mathrm{FER}=10^{-3}$, by log-linear interpolation, and apply this rule to all curves. Table~\ref{tab:ladder} summarises the results. The code for the encoder, the decoder and all simulations is publicly available.\footnote{\url{https://github.com/frederikwalter/Orthogonal-LDPC-Codes-for-the-IDS-Channel-ITW-2027-}}

\begin{table}[t]
    \centering
    \caption{Rate ladder at the waterfall ($\mathrm{FER} = 10^{-3}$) of the
    inter-strand scheme, $L = 100$, in bit/nt.}
    \label{tab:ladder}
    \begin{tabular}{cccccc}
    \toprule
    $\psub$ & $\pins = \pdel$ & $\mathrm{GMI}_{100}$ & $R$ & $I_{100}$ & $I_{100} - R$ \\
    \midrule
    $0$    & $0.110$ & $0.22$ & $0.48$ & $0.94$ & $0.46$ \\
    $0.05$ & $0.082$ & $0.24$ & $0.48$ & $0.91$ & $0.43$ \\
    $0.10$ & $0.055$ & $0.27$ & $0.48$ & $0.89$ & $0.41$ \\
    \bottomrule
    \end{tabular}
\end{table}

\subsection{Gain of Joint Synchronisation}
The separated receiver, which is our receiver without feedback, can transmit $0.48$~bit/nt only up to $p_\mathrm{i}=p_\mathrm{d}=0.038$, $0.029$ and $0.021$ for $p_\mathrm{s}=0$, $0.05$ and $0.10$, even with an ideal outer code. The joint decoder reaches $0.110$, $0.082$ and $0.055$, i.e., it tolerates almost three times the insertion and deletion rate. Equivalently, at our waterfalls the separated receiver supports only $0.22$--$0.27$~bit/nt, about half of our rate. In the terms of Section~\ref{sec:localrates}, this is the part of the gap between $\mathrm{GMI}_L$ and $I_L$ that the anchors recover.

\subsection{Comparison with Concatenated Codes}
\begin{figure}[t]
    \centering
    \vspace{3pt}

\begin{tikzpicture}
\begin{axis}[
    width=\columnwidth, height=0.7\columnwidth,
    xlabel={$p_{\mathrm{i}} = p_{\mathrm{d}}$},
    ylabel={FER},
    xmin=0.025, xmax=0.195,
    ymin=5e-6, ymax=2,
    ymode=log,
    grid=major,
    grid style={dashed, gray!40},
    xtick={0.04, 0.06, 0.08, 0.10, 0.12, 0.14, 0.16, 0.18},
    xticklabel style={/pgf/number format/fixed, /pgf/number format/fixed zerofill,
                        /pgf/number format/precision=2},
    legend cell align={left},
    legend style={font=\footnotesize, legend columns=2, at={(0.5,-0.2)},
                  anchor=north, draw=none, /tikz/column 2/.style={column sep=6pt}},
    unbounded coords=discard,
    every axis plot/.append style={thick},
]
    \addplot[color=PsZero, mark=*]
        table[x=p, y=fer_ps000, col sep=comma]{fig/fer_sc_ldpc_4ary.csv};
    \addlegendentry{this work, $p_{\mathrm{s}} = 0$}
    \addplot[color=PsZero, mark=o, dashed, mark options={solid}]
        table[x=p, y=maarouf_tit23_ps000, col sep=comma]{fig/fer_sc_ldpc_4ary.csv};
    \addlegendentry{one-pass, $p_{\mathrm{s}} = 0$}
    \addplot[color=PsFive, mark=square*]
        table[x=p, y=fer_ps050, col sep=comma]{fig/fer_sc_ldpc_4ary.csv};
    \addlegendentry{this work, $p_{\mathrm{s}} = 0.05$}
    \addplot[color=PsFive, mark=square, dashed, mark options={solid}]
        table[x=p, y=maarouf_tit23_ps050, col sep=comma]{fig/fer_sc_ldpc_4ary.csv};
    \addlegendentry{one-pass, $p_{\mathrm{s}} = 0.05$}
    \addplot[color=PsTen, mark=triangle*]
        table[x=p, y=fer_ps100, col sep=comma]{fig/fer_sc_ldpc_4ary.csv};
    \addlegendentry{this work, $p_{\mathrm{s}} = 0.10$}
    \addplot[color=PsTen, mark=triangle, dashed, mark options={solid}]
        table[x=p, y=maarouf_tit23_ps100, col sep=comma]{fig/fer_sc_ldpc_4ary.csv};
    \addlegendentry{one-pass, $p_{\mathrm{s}} = 0.10$}
    \addplot[color=black, mark=diamond, dashed, mark options={solid}]
        table[x=p, y=sim_n960_tvc1, col sep=comma]{fig/fer_maarouf_istc23.csv};
    \addlegendentry{iterative, $N = 960$}
    \addplot[color=black, dashdotted, no markers]
        table[x=p, y=dt_n960_tvc1, col sep=comma]{fig/fer_maarouf_istc23.csv};
    \addlegendentry{DT bound, $N = 960$}
    \addplot[color=black!50, mark=pentagon, dashed, mark options={solid}]
        table[x=p, y=sim_n128_tvc1, col sep=comma]{fig/fer_maarouf_istc23.csv};
    \addlegendentry{iterative, $N = 128$}
    \addplot[color=black!50, dashdotted, no markers]
        table[x=p, y=dt_n128_tvc1, col sep=comma]{fig/fer_maarouf_istc23.csv};
    \addlegendentry{DT bound, $N = 128$}
    \addplot[color=PsZero, densely dotted, forget plot]
        coordinates {(0.10, 1e-4) (0.11, 1.1e-3)};
    \addplot[color=PsZero, only marks, mark=*, forget plot]
        table[x=p, y=ub_ps000, col sep=comma]{fig/fer_sc_ldpc_4ary.csv};
    \addplot[color=PsZero, ->, forget plot]
        coordinates {(0.10, 1e-4) (0.10, 1.5e-5)};
    \addplot[color=PsFive, densely dotted, forget plot]
        coordinates {(0.07, 1e-4) (0.08, 6e-4)};
    \addplot[color=PsFive, only marks, mark=square*, forget plot]
        table[x=p, y=ub_ps050, col sep=comma]{fig/fer_sc_ldpc_4ary.csv};
    \addplot[color=PsFive, ->, forget plot]
        coordinates {(0.07, 1e-4) (0.07, 1.5e-5)};
    \addplot[color=PsTen, densely dotted, forget plot]
        coordinates {(0.04, 1e-4) (0.05, 1e-4)};
    \addplot[color=PsTen, only marks, mark=triangle*, forget plot]
        table[x=p, y=ub_ps100, col sep=comma]{fig/fer_sc_ldpc_4ary.csv};
    \addplot[color=PsTen, ->, forget plot]
        coordinates {(0.04, 1e-4) (0.04, 1.5e-5)};
\end{axis}
\end{tikzpicture}
    \vspace{-7mm}
    \caption{FER versus $\pins = \pdel$ of the inter-strand scheme ($R = 0.48$
    bit/nt; solid, filled) and of concatenated codes at $R = 1/2$ bit/nt
    (dashed, open): one-pass with
    $N = 960$~\cite{maarouf_2023_ConcatenatedCodes}; iterative with
    $N = 960$ (black) and $N = 128$ (grey), with their dependence-testing (DT)
    bounds (dash-dotted), $\psub = 0$ only~\cite{maarouf_2023_FiniteBlocklength}. Colour marks
    $\psub$. Arrows: no frame error in $10^4$ frames, drawn at $10^{-4}$.}
    \vspace{-5mm}
    \label{fig:fer4}
\end{figure}
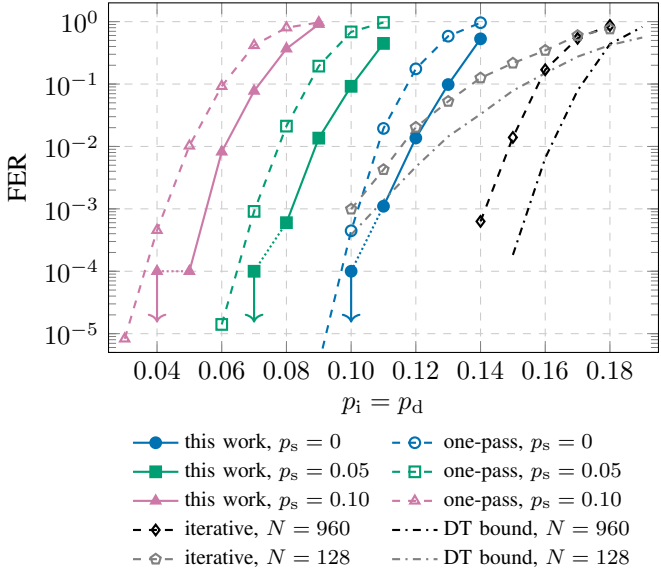
Fig.~\ref{fig:fer4} compares our scheme with the concatenated construction of~\cite{maarouf_2023_ConcatenatedCodes}, the closest rate-matched point in the literature: the same alphabet and channel, $R=1/2$~bit/nt and a single read per strand. Of the $1.5$~bit/nt it gives up against the quaternary alphabet, $1$~bit/nt goes to the inner synchronisation code and $0.5$ to the outer code. Ours spends none and decodes at higher insertion and deletion rates at every substitution rate, with margins of $0.8$ to $1.3$ percentage points. The margin grows with $p_\mathrm{s}$, which suggests that joint synchronisation across strands degrades more gracefully than a single-strand inner code when the observations are noisier. No frame error occurred in $10^4$ frames at $p_\mathrm{i}=p_\mathrm{d}=0.10$, $0.07$ and $0.04$, respectively, and no error floor is visible down to $\mathrm{FER}=10^{-4}$.

With an EXIT-optimised outer code and up to $100$ turbo iterations, the same construction reaches $0.142$ at $N=960$~\cite{maarouf_2023_FiniteBlocklength}. These ingredients moved that family from $0.102$ to $0.142$, and none is used in our decoder yet. At $N=128$, the length closest to our strands, the iterative scheme reaches $0.100$, and its DT bound is $0.107$; our scheme decodes at $0.110$, since its codeword length is not tied to the strand length. Three differences should be kept in mind: our codeword spans the whole frame of $38\,400$ nucleotides, whereas the baselines decode codewords of $N$ nucleotides; the baseline of~\cite{maarouf_2023_ConcatenatedCodes} decodes in a single pass; and~\cite{maarouf_2023_FiniteBlocklength} reports $p_\mathrm{s}=0$ only.

Our decoder also needs fewer operations than the iterative concatenated decoder. Following~\cite{maarouf_2023_ConcatenatedCodes}, we count the edges of the inner trellis: a forward--backward pass costs $2^{k_\mathrm{i}}D\,(n_\mathrm{i}(I_{\max}+1)+1)/n_\mathrm{i}$ edges per nucleotide for $D$ drift states, inner codewords of $n_\mathrm{i}$ nucleotides carrying $k_\mathrm{i}$ bits, and at most $I_{\max}$ insertions per symbol. The concatenated decoder of~\cite{maarouf_2023_FiniteBlocklength} limits the drift to $\pm5\sqrt{Np_\mathrm{d}/(1-p_\mathrm{d})}\approx\pm63$, i.e., $D=127$ at $N=960$ and $p_\mathrm{d}=0.142$, with $n_\mathrm{i}=k_\mathrm{i}=4$ and $I_{\max}=2$, so a pass costs $16\cdot127\cdot13/4\approx6.6\cdot10^3$ edges per nucleotide, repeated in up to $100$ turbo iterations. Our trellis has $n_\mathrm{i}=1$, $k_\mathrm{i}=2$ and $I_{\max}=3$, and $D$ is the beam width $B=70$, which does not grow with $L$, so a full-strand pass costs $4\cdot70\cdot5=1.4\cdot10^3$ edges per nucleotide. At $p_\mathrm{i}=p_\mathrm{d}=0.11$, decoding a frame takes on average $75$ steps of two windows of $W=9$ positions, about $14$ full-strand passes, plus an initial pass from both ends that counts as two, i.e., about $16$ passes or $2.2\cdot10^4$ edges per nucleotide. A complete frame decoding thus costs about as much as three turbo iterations of the concatenated decoder.

\subsection{Distance to Achievable Rate}
At its waterfalls our scheme operates $0.41$--$0.46$~bit/nt below $I_{100}$, which would permit $0.48$~bit/nt up to $0.194$, $0.168$ and $0.139$. Since no rate is spent on synchronisation, this gap is not built into the construction. In the frame logs, failing frames end with both waves stalled and a contiguous unfixed gap, not with a diffuse error pattern: a position that neither wave resolves blocks the span between them, and the release of Section~\ref{sec:schedule} acts only after both waves have stalled. The gap is therefore a matter of the schedule and the outer code, which is unoptimised, rather than of the construction.

\section{Conclusion}
\label{sec:conclusion}
We proposed an inter-strand SC-LDPC scheme for the quaternary IDS channel whose
coupling axis coincides with the position along the strands, so that every
position the code certifies becomes a synchronisation anchor for all strands at
once. Without an inner code or markers, it operates at $R = 0.48$ bit/nt to the
right of a rate-matched concatenated construction at every substitution rate,
and beyond a single-strand code of $N = 128$ nucleotides and its DT bound,
because the codeword length no longer depends on the strand length. At its
waterfall it transmits about twice the rate of a separated receiver, but it
remains $0.41$--$0.46$ bit/nt below the information rate of a strand.

Future work will target this gap through an SC-LDPC code optimised by
EXIT analysis, larger coupling widths, and a permutation of the strands at each
spatial position. Further directions are longer simulations to test for an
error floor below $\mathrm{FER} = 10^{-4}$, multiple reads per strand, soft
input from a nanopore basecaller, and an extension to the multi-strand channel.

\clearpage
\section*{Acknowledgement}
The authors used the AI system Claude (Anthropic) to assist in improving the code for the simulations and to generate the figures in this manuscript.
All AI-generated content was checked for accuracy and revised by the authors, who take full responsibility for the content of this paper.

\bibliographystyle{IEEEtran}
\bibliography{LDPC_deletion}

\end{document}